\documentclass[a4paper,USenglish]{lipics}

\usepackage{enumerate}

\usepackage{tikz}

\newcommand{\var}{\mathtt{var}}
\newcommand{\sat}{\mathtt{sat}}

\newcommand{\cla}{\mathtt{cla}}
\newcommand{\bigoh}{\mathcal{O}}

\theoremstyle{definition}
\newtheorem{observation}[theorem]{Observation}

\title{On Satisfiability Problems with a Linear Structure}

\author[1,2]{Serge Gaspers}
\author[3]{Christos Papadimitriou}
\author[4]{Sigve Hortemo S\ae ther}
\author[4]{Jan Arne Telle}
\affil[1]{UNSW, Sydney, Australia\\sergeg@cse.unsw.edu.au}
\affil[2]{Data61 (formerly NICTA), Sydney, Australia}
\affil[3]{UC Berkeley, USA\\christos@berkeley.edu}
\affil[4]{University of Bergen, Norway\\sigve.sether@ii.uib.no, telle@ii.uib.no}

\authorrunning{S.Gaspers, C.Papadimitriou, S.H.S\ae ther, J.A.Telle } 
  
\Copyright{Serge Gaspers, Christos Papadimitriou, Sigve Hortemo S\ae ther, Jan Arne Telle}

\begin{document}

\maketitle

\begin{abstract}
It was recently shown \cite{STV} that satisfiability is polynomially solvable  when the incidence graph is an interval bipartite graph (an interval graph turned into a bipartite graph by omitting all edges within each partite set).  Here we relax this condition in several directions:  First, we show that it holds for $k$-interval bigraphs, bipartite graphs which can be converted to interval bipartite graphs by adding to each node of one side at most $k$ edges; the same result holds for the counting and the weighted maximization version of satisfiability.  Second, given two linear orders, one for the variables and one for the clauses, we show how to find, in polynomial time, the smallest $k$ such that there is a $k$-interval bigraph compatible with these two orders. On the negative side we prove that, barring complexity collapses, no such extensions are possible for CSPs more general than satisfiability. 
We also show NP-hardness of recognizing 1-interval bigraphs.
\end{abstract}


\section{Introduction}
Constraint satisfaction problems (CSPs) such as satisfiability are both ubiquitous and difficult to solve. It is therefore essential to identify and exploit any special structure of instances that make CSPs susceptible to  algorithmic techniques.  One large class of such structured instances comprises CSPs whose constraints can be arranged in a linear manner, presumably reflecting temporal or spatial ordering of the real-life problem being modeled.  A well known example is the {\em car sequencing} class of CSPs proposed by the French automobile manufacturer Renault in 2005 and reviewed in \cite{Solnon}.  

But defining what it means for a CSP to have ``a linear structure'' is not straightforward. The linear structure should be reflected in the incidence graph of the instance, but how exactly? 
Previous work has considered satisfiability instances with incidence graphs of bounded tree-width or bounded clique-width  \cite{DBLP:journals/dam/FischerMR08,DBLP:conf/stacs/PaulusmaSS13,DBLP:journals/jda/SamerS10,DBLP:conf/isaac/SlivovskyS13,DBLP:conf/sat/Szeider03}.
Instances that are in some sense close to efficiently solvable instances have been studied in terms of backdoors \cite{DBLP:conf/birthday/GaspersS12,DBLP:conf/ijcai/WilliamsGS03}, in particular for CNF formulas that have a small number of variables whose instantiations give formulas of bounded treewidth \cite{DBLP:conf/focs/GaspersS13}. An important special case of bounded tree-width is bounded path-width, a measure of how path-like a graph is and a strong indication of linear structure. Bounded clique-width is a stronger notion, in which the graph's cliques have a linear structure.

Recently another direction for defining linear structure in CSPs was proposed, based in a time honored graph-theoretic conception of linear structure: {\em interval graphs,} the intersection graphs of intervals on the line.  Interval graphs are a well-known class of graphs, going back to the 1950s, used to model temporal reasoning \cite{Golumbic:1993:CAR:174147.169675}, e.g. in resource allocation and scheduling \cite{Bar-Noy:2000:UAA:335305.335410}. However, the incidence graphs we care about are bipartite, and the only connected interval graphs that are bipartite are trees.  A bipartite version of interval graphs was introduced by Harary et al.~in 1982 \cite{harary1982bipartite}: An {\em interval bigraph} is, informally, a bipartite graph\footnote{We use ``bigraph'' and ``bipartite graph'' interchangeably.} in which each vertex is associated with an interval, and there is an edge between two vertices {\em on different sides} if and only if the corresponding intervals intersect.   Interval bigraphs form a natural and fairly rich class of bipartite graphs, containing, {\sl e.g.,} all bipartite permutation graphs, which have been shown to have unbounded clique-width and thus also unbounded treewidth or pathwidth \cite{DBLP:journals/arscom/BrandstadtL03}. 

Interval bigraphs have been studied quite extensively, and several important facts are known about them.  First, they can be recognized in polynomial time:  In 1997 
M{\"u}ller gave an algorithm with running time $\bigoh(|V||E|^6(|V|+|E|)\log |V|)$ \cite{DBLP:journals/dam/Muller97}, and a 2012 technical report \cite{DBLP:journals/corr/abs-1211-2662} gives an algorithm with running time $\bigoh(|V|(|E|+|V|))$.  Importantly, Hell and Huang \cite{HellHuang} gave in 2004 a useful alternative characterization of interval bigraphs  as all bipartite graphs whose set of vertices can be ordered on the line so that the set of neighbors of each vertex coincides with an interval whose high end is the position of the vertex (see Lemma  \ref{lemma} below for the formal statement).

Interval bigraphs constitute a natural basis for identifying an important class of CSPs possessing a linear order:  Define {\em an interval CSP} as a CSP whose variable-constraint incidence graph is an interval bigraph.  In \cite{STV}, a general dynamic programming approach to solving CSPs was developed, and one consequence of that framework is that satisfiability --- even weighted {\small MAXSAT} and {\small $\#$SAT} --- on interval CNF formulae with $m$ clauses and $n$ variables can be solved in time $\bigoh(m^3(m+n))$ (stated as Theorem \ref{stv} below). See also \cite{Brault-Baron}.

The present work is about extending this result, in several natural directions:

\begin{enumerate}
\item  Many CSPs are not interval CSPs.  Can the definition of interval CSPs be extended usefully, so that a limited number of ``faults'' in the interval structure of CSPs is tolerated by polynomial time algorithms?
\item A second question is, suppose that an instance of satisfiability has no natural overall linear order over the union of its variables and constraints, but we know a natural linear order for the constraints, and another natural linear order for the variables.  Under what circumstances is it possible to merge these two linear orders into one, so that the resulting bipartite graph is an interval bigraph?  
\item If an order as in (2) above does not exist, can at least a merged order be found so the resulting bipartite graph is as close as possible (presumably in some algorithmically useful sense as in (1) above) to an interval bigraph?
\item  Finally, to what extent can these algorithmic results be extended to CSPs beyond satisfiability?
\end{enumerate}

In this paper we address and largely resolve these questions.  In particular, our contributions are the following:
\begin{enumerate}
\item  We define a useful measure of how much the incidence graph of a CSP instance differs from an interval bigraph:  The smallest number $k$ such that the incidence bigraph becomes interval if each constraint vertex of the bigraph has at most $k$ edges added to it. Deciding if $k \leq 1$ is NP-hard (Theorem \ref{sec}) but we show (Theorem \ref{first}) that given an ordering certifying a value of $k$ such instances of satisfiability with $m$ clauses and $n$ variables can be solved in $\bigoh(m^34^k(m+n))$ time.  Ditto for {\small MAXSAT} and {\small $\#$SAT}; the exponential dependence on $k$ is, of course, expected.
\item  We give a simple characterization of when two linear orders, one for constraints and one for variables, can be merged so that the resulting total order satisfies the Hell-Huang characterization of interval bigraphs.  
\item  We also show that, if no such merging is possible, we can find in polynomial time --- through a greedy algorithm --- the minimum $k$ such that the incidence graph becomes interval with the addition of at most $k$ edges to each constraint vertex.  Hence, in the case of satisfiability, if this minimum is bounded then polynomial algorithms result.
\item  Finally, we show that the approach in (1) above --- which started us down this path ---   does not work for general CSPs, in that CSP satisfiability is intractable even when the incidence graph has the same favorable structure as in (1), with bounded $k$ (Theorem \ref{hard}). 
\end{enumerate}

\subsection*{Definitions and Background}
Since we mostly deal with satisfiability, we denote our bipartite graphs as $G=(\cla, \var, E)$, where $\cla$ stands for clauses and $\var$ for variables.
\begin{definition}
A bipartite graph $G=(\cla, \var, E)$ is an {\em interval bigraph} if every vertex can be assigned an interval on the real line such that for all $x \in \var$ and $c \in \cla$ we have $xc \in E$ if and only if the corresponding intervals intersect.  
A Boolean formula in conjunctive normal form (CNF) is called an {\em interval CNF formula} if the corresponding incidence graph ($\cla$ the clauses, $\var$ the variables, $E$ the incidences) is an interval bigraph.
\end{definition} 

A most interesting alternative characterization of interval bigraphs by Hell and Huang \cite{HellHuang} is stated here, expressed in terms of interval CNF formulas.
\begin{lemma}\label{lemma}\cite{HellHuang}
A CNF formula is an interval CNF formula if and only if its variables and clauses can be totally ordered (indicated by $<$) such that for any variable $x$ appearing in a clause $C$:
\begin{itemize}
 \item[1.] if $x'$ is a variable and $x < x' < C$ then $x'$ also appears in $C$, and
 \item[2.] if $C'$ is a clause and $C < C' < x$ then $x$ also appears in $C'$.
\end{itemize}
\end{lemma}
We call an ordering of the variables and clauses of an interval CNF formula satisfying the lemma an {\em interval ordering}.  Interval bigraphs can be recognized in polynomial time \cite{DBLP:journals/dam/Muller97}, see also \cite{DBLP:journals/corr/abs-1211-2662}.

\section{k-interval Bigraphs} 
Recent work has articulated efficient algorithms in the dynamic programming style for interval CNF formulae.

\begin{theorem}\label{stv} \cite{STV}
Given an interval CNF formula on $n$ variables and $m$ clauses and an interval ordering of it, $\#$SAT and weighted {\sc MaxSAT} can be solved in time $\bigoh(m^3(m+n))$.
\end{theorem}

Combining Theorem \ref{stv} with the recognition algorithm of \cite{DBLP:journals/dam/Muller97} gives the following:

\begin{corollary}
Given a CNF formula, it can be decided if it is an an interval CNF formula, and if so $\#$SAT and weighted {\sc MaxSAT} can be solved in polynomial time.
\end{corollary}

We want to generalize this result to a larger class of formulae. To this end we introduce the following graph classes and formula classes, parametrized by $k\geq 1$.

\begin{definition}\label{def-kint}
A bipartite graph $G=(\cla, \var, E)$ is a {\em $k$-interval bigraph}
if we can add at most $k$ edges to each vertex in $\cla$ such that the resulting bipartite graph is an interval bigraph.   A CNF formula is called a {\em $k$-interval CNF formula} if its incidence graph (with clause vertices being $\cla$) is a $k$-interval bigraph.
\end{definition}

Note that 0-interval bigraphs are the interval bigraphs, and 1-interval bigraphs allow as many exceptions (added edges) as there are clauses. 
Unfortunately, the recognition problem for $k$-interval bigraphs becomes hard, already when $k=1$.
The proof is by reduction from the strongly NP-hard 3-\textsc{Partition} problem and is given in Section 5.

\begin{theorem}\label{sec}
Given a bipartite graph $G$ and an integer $k$, deciding if $G$ is a $k$-interval bigraph is NP-hard, even when $k=1$.
\end{theorem}

The alternative characterization of Lemma \ref{lemma} can be extended to $k$-interval bigraphs.

\begin{lemma}\label{key}
 A CNF formula is a  $k$-interval CNF formula if and only if its variables and clauses can be totally ordered such that for any clause $C$ there are at most $k$ variables $x$ not appearing in $C$ where either
\begin{itemize}
	\item[1.] a variable $x'$ appears in $C$ with $x' < x < C$, or
	\item[2.] $x$ appears in a clause $C'$ with $C' < C < x$.
\end{itemize}
\end{lemma}

\begin{proof}
The lemma follows directly from Definition \ref{def-kint} 
and Lemma \ref{lemma}. 
\end{proof}

\begin{definition}
For a $k$-interval CNF formula we call a total ordering of the kind guaranteed by Lemma \ref{key} a $k$-interval ordering.
\end{definition}

Our basic algorithmic result is that, given a $k$-interval ordering of a $k$-interval CNF formula, \#SAT and {\sc MaxSAT} can be solved via a fixed-parameter tractable (FPT, see \cite{DowneyF}) algorithm parameterized by $k$.

\begin{theorem}\label{first}
Given a CNF formula and a $k$-interval ordering of it, we solve \#SAT and weighted {\sc MaxSAT} in time $\bigoh(m^34^k(m+n))$.
\end{theorem}

\begin{proof}
The full proof for \#SAT and weighted {\sc MaxSAT} is given in Section 4;  here we give a straightforward construction establishing a weaker result for satisfiability only.
  
The basic observation is that the satisfiability of a CNF formula is not affected if a clause $C$ is replaced by a particular set of clauses, defined next.  Fix any set of $\ell\geq 0$ variables not occurring in $C$, and replace $C$ with the $2^{\ell}$ clauses of the form $(C\lor D_j): j=1,\ldots,2^{\ell}$, where $D_j$ ranges over the $2^{\ell}$ possible clauses containing the $\ell$ variables.  It is easy to see that a truth assignment satisfies the new formula if and only if it satisfied the original one.  It is further clear that the satisfiability of the formula remains unaffected if all clauses are so replaced, for different sets of variables and $\ell \geq 0$.  Finally, if a CNF formula is $k$-interval, then it has such an equivalent variant whose incidence graph is an interval bigraph.  An FPT algorithm (albeit with running time $\bigoh(m^3 8^k(m 2^k +n))$ instead of $\bigoh(m^34^k(m+n))$) results.
\end{proof}

We next show that the $k$-interval structure is not helpful for general CSPs:
\begin{theorem}\label{hard}
 Given a CSP instance $I$ with variable-constraint incidence graph $G$ and an interval bigraph $G'$ obtained from $G$ by adding at most
 $k$ edges to each constraint vertex, deciding the satisfiability of $I$ is $W[1]$-hard parameterized by $k$. 
\end{theorem}

\begin{proof}
 It is known that CSP is $W[1]$-hard parameterized by the number of variables \cite{PapadimitriouY99}.
 Given a CSP instance with $k$ variables, we can turn its incidence graph into an interval bigraph by adding all possible edges between variables and constraints. This creates a complete bipartite graph and adds at most $k$ edges to each constraint vertex. 
\end{proof}

\section{Merging Linear Orders}

Theorem  \ref{hard} tells us that our ambition for new algorithmic results based on the concept of k-interval bigraph should be limited to CSPs of the satisfiability kind, while Theorem  \ref{sec} suggests that the new concept of k-interval bigraph can only extend the class of solvable problems either in special cases, or indirectly, in specific contexts.  In this section we derive an algorithmic result of the latter type.

Suppose that the real life situation modelled by the CNF formula has linearly ordered clauses, and linearly ordered variables, but there is no readily available linear order for both. That is, we assume the input comes with two linear orderings, one for the variables and one for the clauses.  We wish to find the minimum value of $k$ such that there exists a $k$-interval ordering compatible with both. 

\bigskip

\noindent
{\bf Problem}: {\sc  Merging to minimum $k$-interval bigraph ordering}\\
{\bf Input}: Bipartite graph $G=(\cla,\var,E)$, a total order of $\cla$, and a total order of $\var$\\
{\bf Output}: The minimum $k$ such that we can merge the two orders into a $k$-interval ordering of $\cla \cup \var$.

\bigskip
\noindent
Consider first the case $k=0$.

\begin{figure}[tb]
      \center
      \includegraphics[scale=0.3]{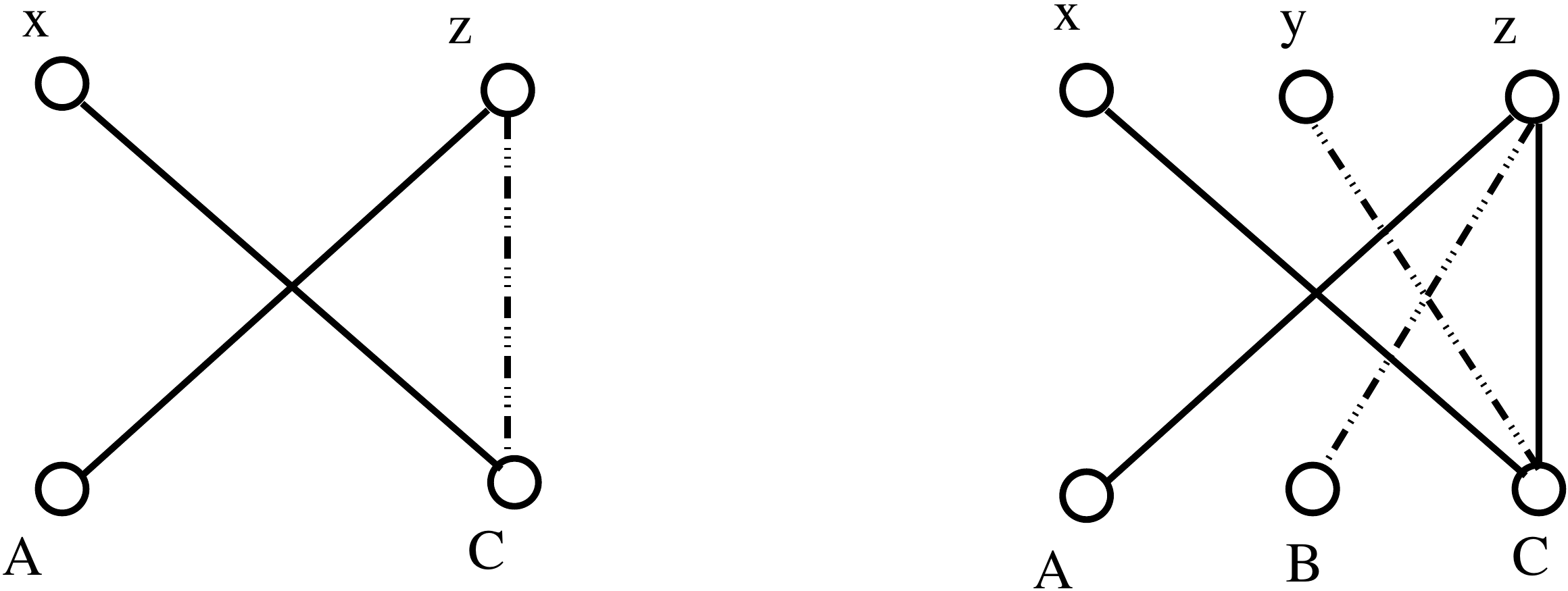}
      \caption{\label{fig:tree}Obstructions to merging into an interval bigraph ordering: variables ordered $x < y < z$, clauses $A < B < C$, with solid lines indicating edges of the incidence graph and dotted lines indicating non-edges, with remaining possibilities being any combination of edges or non-edges.}
\end{figure}

\begin{lemma}\label{cl}
If a formula is given with variable ordering, clause ordering, and incidences containing one of the obstructions in Figure \ref{fig:tree}, then it cannot be merged into an interval bigraph ordering.
\end{lemma}

\begin{proof}
Consider the left-hand obstruction in Figure \ref{fig:tree}. We cannot insert $z$ after $C$, since we get $A <  C < z$ violating Condition 2 in Lemma \ref{lemma}.
On the other hand, we cannot insert $z$ before $C$, since we get $x <  z < C$ violating Condition 1 in Lemma \ref{lemma}. 

Consider the right-hand obstruction in Figure \ref{fig:tree}. We cannot insert $z$ after $B$, since we get $A <  B < z$ violating Condition 2 in Lemma \ref{lemma}.
On the other hand, we cannot insert $y$ before $C$, since we get $x <  y < C$ violating Condition 1 in Lemma \ref{lemma}. Thus, since $B < C$ this leaves no place to insert $z$ without violating Lemma \ref{lemma}.
\end{proof}

It turns out that, if there are no obstructions as in Figure \ref{fig:tree} then {\sc Merging to minimum $k$-interval bigraph ordering} has a solution with $k=0$. Thus, for any instance where the solution has value $k>0$ we can view the task as one of iteratively adding edges until the result has no obstruction as in Figure \ref{fig:tree}. 
On the face of it this is non-trivial, as there is more than one way of fixing an obstruction, with varying edge costs, and some ways may lead to a new obstruction appearing. For an example of this see Figures \ref{fig:obstruct} and \ref{fig:fix-obstruct}.

    
\begin{figure}[tb]
      \center
      \includegraphics[scale=0.3]{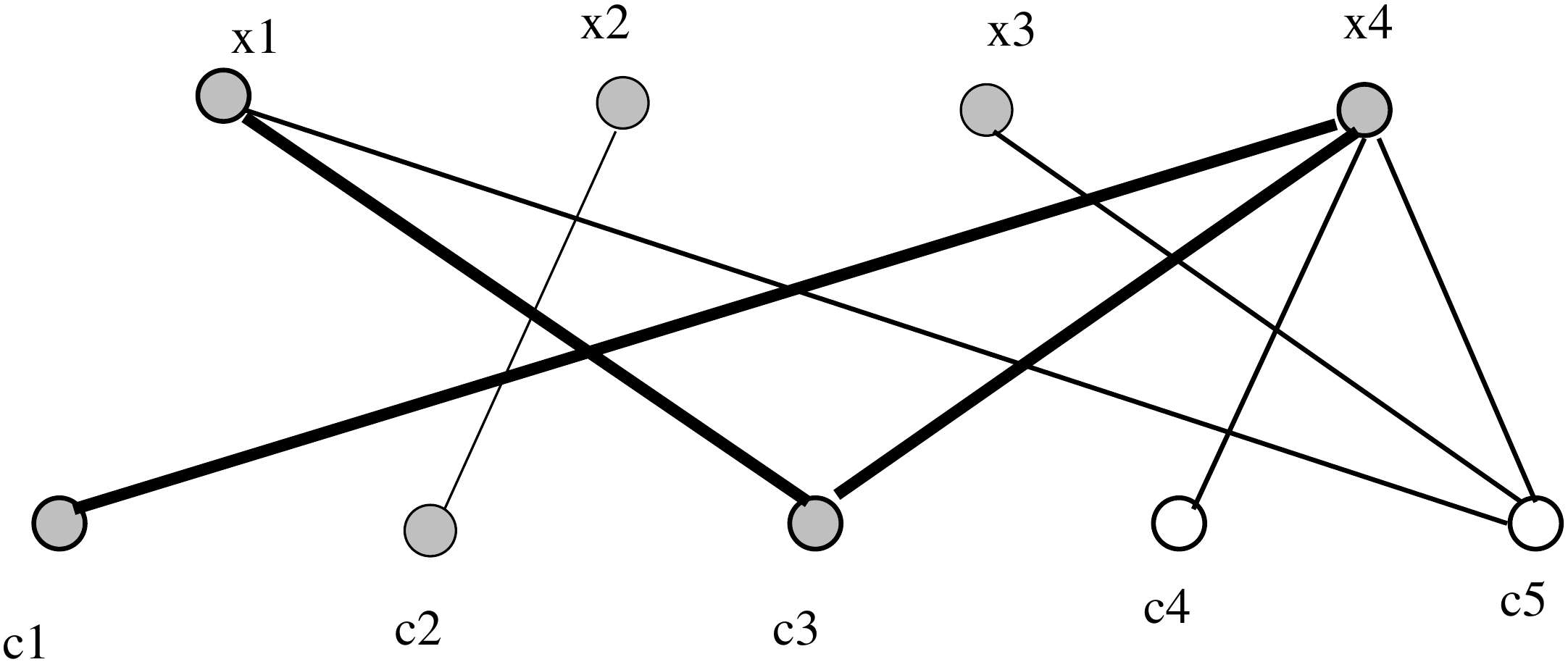}
     \caption{\label{fig:obstruct}Consider the above input, with non-incidences indicated by non-edges. The bold edges and gray nodes show two overlapping obstructions as on the right side of Figure \ref{fig:tree}. Applying Lemma \ref{lemma} these obstructions can be fixed in at least two ways: adding edge $c2x4$ by positioning $c3 < x2$; or adding edges $c3x2$ and $c3x3$ by positioning $c2 > x4$. In this last case a new obstruction appears, see Figure \ref{fig:fix-obstruct}.   
     }
\end{figure}
\begin{figure}[tb]
	 \center
     \includegraphics[scale=0.3]{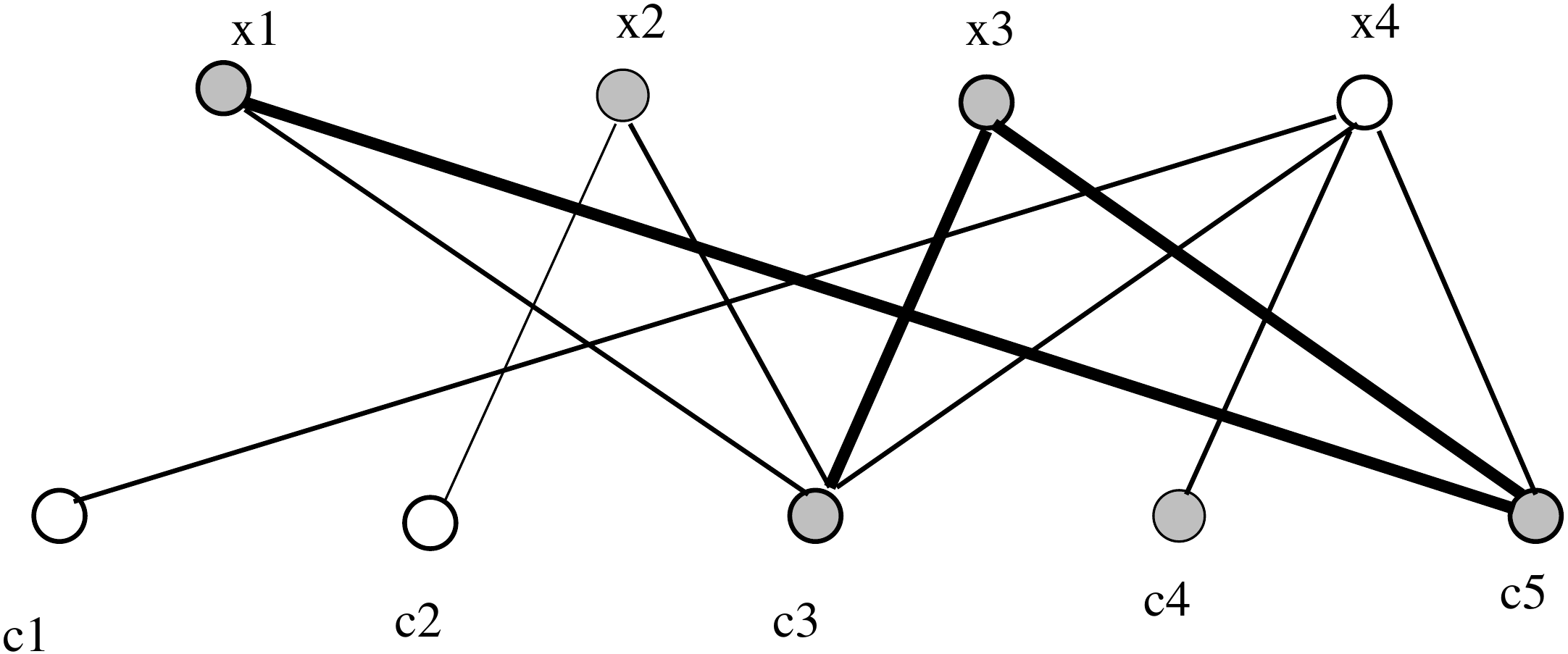}
     \caption{\label{fig:fix-obstruct}Assume we fixed the obstructions from Figure \ref{fig:obstruct} by adding edges $c3x2$ and $c3x3$. We then get a new obstruction based in bold edges and gray nodes. 
     }
\end{figure}
   
   Nevertheless, a greedy approach will efficiently solve {\sc Merging to minimum $k$-interval bigraph ordering}. Let us describe it. Assume the input ordering on variables and clauses is $x_1,...,x_n$ and $c_1,...,c_m$. All orderings we consider will be compatible with these input orderings. The greedy strategy works as follows. Start  with $k=0$ and consider clauses by decreasing index $c_m, c_{m-1},$ etc. Insert $c_i$ among the variables in the highest possible position, below the position of $c_{i+1}$, that does not lead to more than $k$ edges being added to $c_i$. If no such position exists then increase $k$ and start all over again with $c_m$. The correctness of this strategy relies on the following observation.
   
\begin{observation}\label{obs}
For any fixed position of $c_i$ among the variables
the number of edges we must add to clause $c_i$ does not depend on where the other clauses are inserted, as long as $c_1,...,c_{i-1}$ end up below $c_i$ and $c_{i+1},...,c_m$ above $c_i$.
\end{observation}

\begin{proof}
By Lemma \ref{key} we must add to $c_i$ exactly one edge for each variable $x$ not appearing in $c_i$, where $x$ satisfies one of the two conditions stated in Lemma \ref{key}. For the first condition note that $C'$ can be any of $c_1,...,c_{i-1}$ but no other clause. For the second condition note that it does not depend on any other clause, only on the position of $c_i$ among the variables.
\end{proof}

In our greedy strategy, when deciding where to insert $c_i$ the only restriction imposed on us by earlier decisions is that $c_i$ must end up below the position of $c_{i+1}$. To allow the maximum degree of freedom we simply ensure that we have inserted $c_{i+1}$ in the highest possible position. The pseudocode is in Figure \ref{alg1}.

\begin{figure}[h!]
  \center
  \begin{tabular}{|l|}
    \hline
    \textbf{Greedy Algorithm} for merging to minimum $k$-interval bigraph ordering\\
    \hline \hline
    \begin{minipage}{1.0\linewidth}
    \begin{tabbing}
      \=\textbf{output:} \=\kill
      \>\textbf{input:}\>$G=(\cla, \var, E)$, orderings $\cla=c_1,c_2,...,c_m$ and $\var=x_1,x_2,...,x_n$\\
      \>\textbf{output:}\>minimum $k$ such that $\cla$ and $\var$ can be merged into a $k$-interval ordering 
    \end{tabbing}
  \end{minipage}\\\hline \hline
  \begin{minipage}{1.0\linewidth}
    \vspace{2pt}
    \begin{tabbing}
      \=xx\=xx\=xx\=xx\=xx\=\kill
\> \> $q:=-1$;\\
\> \> $success:=false$;\\
\> \> \textbf{while} not $success$\\
\>      \> \> $q:=q+1$;\\
\>      \> \> start with the ordering $x_1,x_2,...,x_n$;\\
\>      \> \> \textbf{for} $i=m$ \textbf{downto} 1\\
\>  \>    \> \> \> insert $c_i$ at the highest position, below $c_{i+1}$, where $EdgesAdded(c_i) \leq q$; \\
\>   \>  \> \> \> \textbf{if} no such position exists for clause $c_i$ then \textbf{break} out of the for loop;\\
\> \> \> \textbf{if} all clauses have been inserted then $success:=true$;\\
      \> \> \textbf{output} $q$;
    \end{tabbing}    
  \end{minipage}\\\hline \hline
  \begin{minipage}{1.0\linewidth}
    \vspace{2pt}
    \begin{tabbing}
      \=xx\=xx\=xx\=xx\=xx\=\kill
      \> \> $EdgesAdded(C)$:= number of variables satisfying one of the conditions of Lemma \ref{key}     
    \end{tabbing}    
  \end{minipage}\\
 \hline
\end{tabular}
\caption{Greedy Algorithm for merging to minimum $k$-interval bigraph ordering} \label{alg1}
\end{figure}

\begin{theorem}\label{new}
 The Greedy Algorithm for { \sc Merging to minimum $k$-interval bigraph ordering} is correct and can be implemented to run in time $\bigoh(|E| \log k)$.
\end{theorem}

\begin{proof}
Let us first argue for correctness.
Consider an iteration of the inner loop that successfully found a position for clause $c_i$ among the variables. For the current value of $q$ it is not possible to insert $c_i$ higher than this position without some $c_j$ needing more than $q$ edges added, for some $i \leq j \leq m$. This is in fact a loop invariant, as we inserted the clauses of higher index in the highest possible positions under exactly this constraint, and by Observation \ref{obs} their position does not influence the number of edges added to other clauses.
Similarly, 
if for some $c_i$ and current value of $q$ we encounter 'no such position exists' then in any ordering of $\cla \cup \var$ compatible with the input orders there will be some $c_j, i \leq j \leq m$ which will need more than $k$ edges added. 
Thus, when the algorithm successfully finds positions for all clauses then the current value of $q$ is the correct answer.

Let us now argue for the running time. For the $\log k$ factor, rather than iterating on $q$ until we succeed, we can search for the minimum $k$ by what is known as galloping search, i.e. try $q$ equal to 1, 2, 4, 8, etc until we succeed for an integer $q$, and then do binary search in the interval $[q/2..q]$.
To decide on positions for the clauses in time $\bigoh(|E|)$, for a fixed $q$, we need several program variables. The pseudocode is in Figure \ref{alg2}. 

\begin{figure}[h!]
  \center
  \begin{tabular}{|l|}
  \hline
    Deciding if we can merge to a $q$-interval ordering, for fixed $q$, in $\bigoh(|E|)$ time\\
    \hline 
    \hline
  \begin{minipage}{1.0\linewidth}
    \vspace{2pt}
    \begin{tabbing}
      \=xx\=xx\=xx\=xx\=xx\=xxxx\=xx\=xx\=xx\=xx\=xx\kill
      \>\> $\forall x \in \var$: $live(x):=$ number of clauses $x$ appears in\\
      \>\>  $\forall c \in \cla$: $var(c):=$ the set of variables in $c$\\
      \>\>\>\>\> \>  $low(c):=i$, lowest $i$ with $x_i \in var(c)$\\
      \>\> $livevar:=0$; \\
      \>\>$t:=n$; \\
      \> \> start with the ordering $x_1,x_2,...,x_n$;\\
      \> \> \textbf{for} $i:=m$ \textbf{downto} 1\\
      \> \> \> $inserted := false$;\\
      \> \> \> \textbf{while} not $inserted$   \hspace{1cm}    /* try to insert $c_i$ after $x_t$ */\\
\>  \>    \> \> \> \textbf{if} $livevar+t-low(c_i)-|var(c_i)| \leq q$ \textbf{then} \\
\>  \>    \> \> \>\> insert $c_i$ after $x_t$;\\
 \>  \>    \> \> \>\> $inserted := true$;\\
 \>  \>   \> \> \> \>  $\forall x_j \in var(c_i): live(x_j) := live(x_j)-1$;\\
\>  \>  \> \> \> \>\> \> \>\>\>\;\textbf{if} $live(x_j)=0$ and $j > t$ \textbf{then} $livevar:=livevar-1$;\\
\>   \>   \> \> \>\textbf{else if} $t=0$ \textbf{then} halt: 'no for this value of $q$';\\
\>   \>   \> \> \>\textbf{else} $t:=t-1$; if $live(x_t) >0$ \textbf{then} $livevar:=livevar+1$;\\
      \> \> 'yes for this value of $q$';
    \end{tabbing}    
  \end{minipage}\\\hline 
\end{tabular}
\caption{Deciding if we can merge to a $q$-interval ordering, for fixed $q$, in $\bigoh(|E|)$ time} \label{alg2}
\end{figure}

We maintain for each $x \in \var$ the value $live(x)$ as the number of live clauses $x$ appears in, where a live clause is one whose position has not been decided yet. Also, we maintain $livevar$ as the number of variables indexed higher than the current $x_t$ and appearing in a live clause. Finally, $var(c_i)$ are the variables in clause $c_i$ and $low(c_i)$ the index of its lowest indexed variable.
The number of edges 
needed for $c_i$ if inserted immediately after $x_t$ is then
$$EdgesAdded(c_i)=livevar+t-low(c_i)-|var(c_i)|$$
This is so since we must add to $c_i$ exactly one edge for each variable $x$ not appearing in $c_i$, where $x$ satisfies one of the two conditions stated in Lemma \ref{key}. The first condition counts the number of variables indexed higher than the current $x_t$ and appearing in some clause indexed lower than $c_i$, i.e. $livevar$, minus the number of variables in $c_i$ of index higher than $t$. The second condition counts the number of variables strictly between $x_{low(c_i)}$ and $x_{t+1}$, i.e. 
$t-low(c_i)$, minus the number of variables in $c_i$ of index $t$ or less. Summing these two counts we get the above.
\end{proof}

\section{Proof of Theorem \ref{first}}
In this section we prove Theorem \ref{first}, namely that if we are given a CNF formula and a $k$-interval ordering of it, we can solve \#SAT and weighted {\sc MaxSAT} in time $\bigoh(m^34^k(m+n))$.
We do this by showing that the input has linear $\mathtt{ps}$-width at most $m2^k+1$ and applying the following result.

\begin{theorem}\label{theorem:SolvingSATonLinearBranchDec}\cite{STV}
  Given a CNF formula $F$ with $n$ variables $\var$ and $m$ clauses $\cla$, 
  and a linear ordering of $\cla \cup \var$ showing that $F$ has linear $\mathtt{ps}$-width at most $p$, we solve
  \textsc{\#SAT} and weighted \textsc{MaxSAT} in
  time $\bigoh(p^2m(m+n))$.
\end{theorem}

We need to clarify what is meant by the linear $\mathtt{ps}$-width of a formula. We start with the related notion of $\mathtt{ps}$-value of a CNF formula $F$ on variables $\var$ and clauses $\cla$. 
For an assignment $\tau$ of 
$\var$, we denote by $\sat(F, \tau)$ the inclusion maximal set $\mathcal{C} \subseteq
\cla$ so that each clause in $\mathcal{C}$ is satisfied by $\tau$.
Such a subset $\mathcal{C} \subseteq
\cla$ is called \emph{projection
  satisfiable}. The $\mathtt{ps}$-value of $F$ is defined to be the number of projection satisfiable subsets of clauses, {\sl i.e.} $|\mathtt{PS}(F)|$, where
\[
  \mathtt{PS}(F) = %
    \{
      \sat(F,\tau)   : 
      \text{$\tau$ is an assignment of $\var$} \} \subseteq 2^{\cla}
    .
\]

Now, consider a linear ordering $e_1,e_2,...,e_{n+m}$ of $\var \cup \cla$. For any $1 \leq i \leq n+m$ we define two disjoint subformulas $F_1(i)$ and $F_2(i)$ crossing the cut between $\{e_1,...,e_i\}$ and $\{e_{i+1},...,e_{n+m}\}$. We define $F_1(i)$ to be the subformula we get by removing from $F$ all clauses not
in $\{e_1,...,e_i\}$ followed by removing from the remaining clauses each literal of a
variable not in $\{e_{i+1},...,e_{n+m}\}$, and we define $F_2(i)$ vice-versa, as the subformula we get by removing from $F$ all clauses not
in $\{e_{i+1},...,e_{n+m}\}$ followed by removing from the remaining clauses each literal of a
variable not in  $\{e_1,...,e_i\}$. 

The $\mathtt{ps}$-width of this linear ordering is defined to be the maximum $\mathtt{ps}$-value over all the $2(n+m)$ subformulas $F_1(1), F_2(1), F_1(2),...,F_2(n+m)$ that cross a cut of the ordering. 
The linear $\mathtt{ps}$-width of $F$ is defined to be the minimum $\mathtt{ps}$-width of all linear orderings of $\var \cup \cla$.

Before giving the lemma that will prove Theorem \ref{first} we state a useful result.

\begin{lemma}\cite{STV}\label{int-cla}
Any interval ordering of an interval CNF formula has $\mathtt{ps}$-width no more than the number of its clauses plus one.
\end{lemma}

\begin{lemma}\label{second}
Let $F$ be a $k$-interval CNF formula on $m$ clauses. Then any $k$-interval ordering of it has $\mathtt{ps}$-width at most $m2^k+1$.
\end{lemma}

\begin{proof}
Starting from $F$ on $m$ clauses and its $k$-interval ordering $\pi$ we first construct an interval CNF formula $F'$ having at most $m2^k$ clauses. Any clause $C$ of $F$ for which Lemma \ref{key} prescribes $k' \leq k$ added edges from $C$ to some $k'$ variables, will be replaced in $F'$ by a set of
$2^{k'}$ clauses consisting of the clause $C$ extended by all linear combinations of these
$k'$ variables.
Note that $F'$ is then an interval CNF formula with an interval ordering $\pi'$ we get from $\pi$ by
naturally expanding a clause $C$ in $\pi$ to the $2^{k'}$ clauses, in any order, that replace $C$ in $F'$.

Applying Lemma \ref{int-cla} all we need to finalize our proof is to show that the $\mathtt{ps}$-width of the $k$-interval ordering $\pi$ of $F$ is no larger than the $\mathtt{ps}$-width of the interval ordering $\pi'$ of $F'$. To do this we must consider cuts of $\pi$.

Consider subformulas $F_1(i)$ and $F_2(i)$ of $F$ crossing a cut of $\pi$. We show that for the corresponding cut in $\pi'$ (i.e. we cut $\pi'$ in the corresponding place without splitting any of the expanded set of clauses) the $\mathtt{ps}$-values of  the subformulas $F_1'$ and $F_2'$ of $F'$ associated with this cut has $\mathtt{ps}$-value no smaller than the $\mathtt{ps}$-values of 
$F_1(i)$ and $F_2(i)$. That is $|\mathtt{PS}(F_1(i))| \leq |\mathtt{PS}(F_1')|$ and $|\mathtt{PS}(F_2(i))| \leq |\mathtt{PS}(F_2')|$. Note that the variables of $F_1(i)$ and $F_1'$ are the same, and similarly the variables of $F_2(i)$ and $F_2'$ are the same. W.l.o.g., we focus on $F_1(i)$ and $F_1'$, which we assume have variables $\var_1$. 

We need to show that if two assignments $a,b$ of $\var_1$ have $\sat(F_1(i), a) \neq \sat(F_1(i), b)$ then also $\sat(F_1', a) \neq \sat(F_1', b)$. W.l.o.g., assume some clause $C \in  \sat(F_1(i), a)$ but $C \not \in  \sat(F_1(i), b)$. We show that we can find a clause $C'$ that distinguishes $a$ and $b$ in $F_1'$ as well. Clause $C$ of $F_1(i)$ comes from an original clause (possibly larger, since $C$ lives across a cut) in $F$.  Assume this original clause was expanded in $F'$ to $2^{k'}$ clauses, for some $k' \leq k$, by extending it with all linear combinations of the
$k'$ new variables. Depending on which variables are on the other side of the cut the clause $C$ of $F_1(i)$ has been expanded to a set of $2^{k''}$, for some $k'' \leq k'$, clauses in $F_1'$, still consisting of all linear combinations of the
$k''$ variables not in $C$.
Since $a$ satisfies $C$ and $C$ is a part of all these expanded clauses we have assignment $a$ satisfying all of them. Since $b$ does not satisfy $C$ there will be exactly one of these $2^{k''}$ clauses that are not satisfied by $b$, namely the one where the linear combination of the new variables is falsified by assignment $b$. This means that $\sat(F_1', a) \neq \sat(F_1', b)$.

Thus the $\mathtt{ps}$-width of the $k$-interval ordering of $F$ is no more than the $\mathtt{ps}$-width of the interval ordering of $F'$ and we are done.
\end{proof}

Combining Theorem \ref{theorem:SolvingSATonLinearBranchDec} with Lemma \ref{second} we arrive at Theorem \ref{first}. Combining with Theorem \ref{new} we get the following.

\begin{corollary}
 Given a CNF formula and two total orderings, one for its $m$ clauses and one for its $n$ variables, we can in polynomial time find the minimum $k$ such that these two orders can be merged into a $k$-interval ordering and then solve \#SAT and MaxSAT in time $\bigoh(m^34^k(m+n))$.
\end{corollary}

\section{Proof of Theorem  \ref{sec}}

In this section we prove Theorem \ref{sec}, namely that it is NP-hard to recognize $k$-interval bigraphs, already for $k=1$.

\begin{proof}
 We give a polynomial time reduction from the 3-\textsc{Partition} problem, which is strongly NP-hard \cite{DBLP:books/fm/GareyJ79}.
 Given an integer $b$, a set $A$ of $3n$ elements, and a positive integer $s(a)$ for each $a\in A$ such that $b/4 < s(a) < b/2$ for each $a\in A$ and $\sum_{a\in A} s(a) = n\cdot b$, the question is whether $A$ can be partitioned into disjoint sets $A_1, \dots, A_n$ such that $\sum_{a\in A_i} s(a) = b$ for each $i\in \{1,\dots, n\}$.
 
 For a 3-\textsc{Partition} instance $(b,A,s)$, we construct an instance $G=(V,E)$ for the 1-interval bigraph recognition problem as follows.
 We assume, w.l.o.g., that $b\ge 4$, and therefore, $s(a)>1$ for each $a\in A$.
 
 We add a set of \emph{slot} vertices $S = \bigcup_{i=1}^n S_i$ with $S_i = \{s_{i,1},\dots,s_{i,b+1}\}$.
 For all $i,j$ with $1\le i\le n$ and $1\le j\le b$ we add a vertex $\ell_{i,j}$ that is adjacent to both $s_{i,j}$ and $s_{i,j+1}$,
 so that
 $(s_{i,1}, \ell_{i,1}, s_{i,2}, \ell_{i,2}, \dots, s_{i,b}, \ell_{i,b}, s_{i,b+1})$ is a path for each $i\in \{1,\dots,n\}$.
 
 For each $i\in \{1,\dots,n-1\}$ we add a \emph{delimiter} vertex $s^d_{i}$,
 and two vertices $\ell^{d,1}_{i}$ and $\ell^{d,2}_i$.
 We make $\ell^{d,1}_{i}$ adjacent to $s_{i,b}$, $s_{i,b+1}$, $s^d_i$, and $s_{i+1,1}$ and
 we make $\ell^{d,2}_{i}$ adjacent to $s_{i,b+1}$, $s^d_i$, $s_{i+1,1}$, and $s_{i+1,2}$.
 The set of delimiter vertices is $D = \bigcup_{i=1}^{n-1} \{s^d_{i}\}$.
 
 We add a \emph{track} vertex $t$ that is adjacent to each vertex in $S\cup D\setminus \{s^d_{1}\}$.
 
 We add \emph{left anchor} vertices
 $a^l$, $\ell^{a,l}$, and
 make $\ell^{a,l}$ adjacent to $a^l$, $s_{1,1}$, and $s_{1,2}$.
 Symmetrically, we add \emph{right anchor} vertices
 $a^r$, $\ell^{a,r}$, and
 make $\ell^{a,r}$ adjacent to $a^r$, $s_{n,b+1}$, and $s_{n,b}$.
 See Figure~\ref{fig:red-recognition} for an illustration of the graph constructed so far.
 
 For each element $a\in A$, we add a \emph{numeral} gadget which is obtained from a path on $2\cdot s(a)+1$ new vertices $(\ell^{n}_{a,0},n_{a,1},\ell^{n}_{a,1},\dots,n_{a,s(a)-1},\ell^{n}_{a,s(a)-1},n_{a,s(a)},\ell^{n}_{a,s(a)})$ and the track vertex $t$ is made adjacent to $n_{a,1},\dots,n_{a,s(a)}$.
 See Figure~\ref{fig:red-numeral} for an illustration of a numeral gadget.

 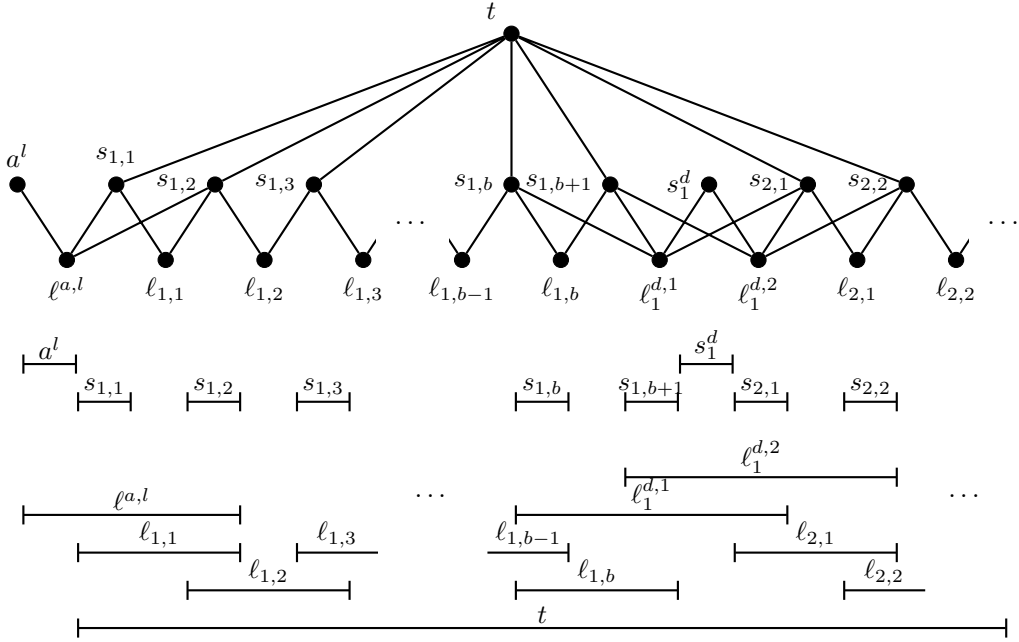
\begin{figure}[tb]
 \begin{center}
 \begin{tikzpicture}[xscale=0.65]
  \tikzset{vertex/.style={minimum size=2mm,circle,fill=black,inner sep=0mm,draw}}
  
  \draw (-2,1) node[vertex,label=above:$a^l$]       (al) {};
  \draw ( 0,1) node[vertex,label=above:$s_{1,1}$]   (s1) {};
  \draw ( 2,1) node[vertex,label=left:$s_{1,2}$]   (s2) {};
  \draw ( 4,1) node[vertex,label=left:$s_{1,3}$]   (s3) {};
  \node at (6,0.5) {$\dots$};
  \draw ( 8,1) node[vertex,label=left:$s_{1,b}$]   (sb) {};
  \draw (10,1) node[vertex,label=left:$s_{1,b+1}$] (sbp) {};
  \draw (12,1) node[vertex,label=left:$s^d_{1}$]   (sd) {};
  \draw (14,1) node[vertex,label=left:$s_{2,1}$]   (s21) {};
  \draw (16,1) node[vertex,label=left:$s_{2,2}$]   (s22) {};
  \node at (18,0.5) {$\dots$};
  
  \draw (-1,0) node[vertex,label=below:$\ell^{a,l}$]     (la) {};
  \draw ( 1,0) node[vertex,label=below:$\ell_{1,1}$]     (l1) {};
  \draw ( 3,0) node[vertex,label=below:$\ell_{1,2}$]     (l2) {};
  \draw ( 5,0) node[vertex,label=below:$\ell_{1,3}$]     (l3) {};
  \draw ( 7,0) node[vertex,label=below:$\ell_{1,b-1}$]   (lbm) {};
  \draw ( 9,0) node[vertex,label=below:$\ell_{1,b}$]     (lb) {};
  \draw (11,0) node[vertex,label=below:$\ell^{d,1}_{1}$] (ld1) {};
  \draw (13,0) node[vertex,label=below:$\ell^{d,2}_{1}$] (ld2) {};
  \draw (15,0) node[vertex,label=below:$\ell_{2,1}$]     (l21) {};
  \draw (17,0) node[vertex,label=below:$\ell_{2,2}$]     (l22) {};
  
  \draw (8,3) node[vertex,label=above left:$t$] (t) {};
  
  \draw[thick] (la)--(s2) (al)--(la)--(s1)--(l1)--(s2)--(l2)--(s3)--(l3);
  \begin{scope}
     \clip (0,0) rectangle (5.25,1);
     \draw[thick] (l3)--(6,1);
  \end{scope}
  \begin{scope}
     \clip (6.75,0) rectangle (7,1);
     \draw[thick] (6,1)--(lbm);
  \end{scope}
  \draw[thick] (lbm)--(sb)--(lb)--(sbp)--(ld1)--(sd)--(ld2)--(s21)--(l21)--(s22)--(l22)
               (sb)--(ld1)--(s21)
               (sbp)--(ld2)--(s22);
  \begin{scope}
    \clip (17,0) rectangle (17.25,1);
    \draw[thick] (l22)--(18,1);
  \end{scope}
  \draw[thick] (t)--(s1) (t)--(s2) (t)--(s3) (t)--(sb) (t)--(sbp) 
               (t)--(s21) (t)--(s22);
 \end{tikzpicture}
 \begin{tikzpicture}[xscale=0.72,yscale=0.5]
  \tikzset{ivl/.style={thick},
  ivl/.default=black,|-|,
  name2/.style={midway,above=-2pt}}

  \draw[ivl] ( 3,6) -- node[name2] {$s_{1,1}$}   ( 4,6);
  \draw[ivl] ( 5,6) -- node[name2] {$s_{1,2}$}   ( 6,6);
  \draw[ivl] ( 7,6) -- node[name2] {$s_{1,3}$}   ( 8,6);
  \draw[ivl] (11,6) -- node[name2] {$s_{1,b}$}   (12,6);
  \draw[ivl] (13,6) -- node[name2] {$s_{1,b+1}$} (14,6);
  \draw[ivl] (14,7) -- node[name2] {$s^d_{1}$}   (15,7);
  \draw[ivl] (15,6) -- node[name2] {$s_{2,1}$}   (16,6);
  \draw[ivl] (17,6) -- node[name2] {$s_{2,2}$}   (18,6);
  
  \draw[ivl]    ( 3  ,0) -- node[name2] {$t$}              (20,0);
  \draw[ivl]    ( 3  ,2) -- node[name2] {$\ell_{1,1}$}      (6,2);
  \draw[ivl]    ( 5  ,1) -- node[name2] {$\ell_{1,2}$}      (8,1);
  \draw[ivl,|-] ( 7  ,2) -- node[name2] {$\ell_{1,3}$}    (8.5,2);
  \node at (9.5,3.5) {$\dots$};
  \draw[ivl,-|] (10.5,2) -- node[name2] {$\ell_{1,b-1}$}   (12,2);
  \draw[ivl]    (11  ,1) -- node[name2] {$\ell_{1,b}$}     (14,1);
  \draw[ivl]    (11  ,3) -- node[name2] {$\ell^{d,1}_{1}$} (16,3);
  \draw[ivl]    (13  ,4) -- node[name2] {$\ell^{d,2}_{1}$} (18,4);
  \draw[ivl]    (15  ,2) -- node[name2] {$\ell_{2,1}$}     (18,2);
  \draw[ivl,|-] (17  ,1) -- node[name2] {$\ell_{2,2}$}   (18.5,1);
  \node at (19.25,3.5) {$\dots$};

  \draw[ivl] (2,7) -- node[name2] {$a^l$}        (3,7);
  \draw[ivl] (2,3) -- node[name2] {$\ell^{a,l}$} (6,3);
 \end{tikzpicture}
 \end{center}
 \caption{\label{fig:red-recognition} A part of the graph constructed by our reduction and a corresponding 1-interval representation formed by the all the vertices except the numeral gadgets. The top two rows of intervals correspond to the vertices in one partite set of the bipartition and the bottom rows to vertices in the other partite set.}
\end{figure}

 \begin{figure}[tb]
 \begin{center}
 \begin{tikzpicture}[baseline=0cm,xscale=0.6]
 \tikzset{vertex/.style={minimum size=2mm,circle,fill=black,inner sep=0mm,draw}}
 
 \draw ( 0,1) node[vertex,label=above:$n_{a,1}$]   (n1) {};
 \draw ( 2,1) node[vertex,label=left:$n_{a,2}$]    (n2) {};
 \draw ( 4,1) node[vertex,label=left:$n_{a,3}$]    (n3) {};
 \node at (5.5,0.5) {$\dots$};
 \draw ( 7,1) node[vertex,label=left:$n_{a,s(a)-1}$]   (nsm) {};
 \draw ( 9,1) node[vertex,label=above:$n_{a,s(a)}$] (ns) {};
 
 \draw (-1,0) node[vertex,label=below:$\ell^n_{a,0}$]      (l0) {};
 \draw ( 1,0) node[vertex,label=below:$\ell^n_{a,1}$]      (l1) {};
 \draw ( 3,0) node[vertex,label=below:$\ell^n_{a,2}$]      (l2) {};
 \draw ( 8,0) node[vertex,label=below:$\ell^n_{a,s(a)-1}$] (lsm) {};
 \draw (10,0) node[vertex,label=below:$\ell^n_{a,s(a)}$]   (ls) {};
 
 \draw (5,3) node[vertex,label=above left:$t$] (t) {};
 
 \draw[thick] (l0)--(n1)--(l1)--(n2)--(l2)--(n3); 
 \begin{scope}
 \clip (0,0) rectangle (4.25,1);
 \draw[thick] (n3)--(5,0);
 \end{scope}
 \begin{scope}
 \clip (6.75,0) rectangle (7,1);
 \draw[thick] (6,0)--(nsm);
 \end{scope}
 \draw[thick] (nsm)--(lsm)--(ns)--(ls);
 \draw[thick] (t)--(n1) (t)--(n2) (t)--(n3) (t)--(nsm) (t)--(ns);
 \end{tikzpicture}\hfill%
 \begin{tikzpicture}[baseline=0.5cm,xscale=0.54,yscale=0.6]
  \tikzset{ivl/.style={thick},
  ivl/.default=black,|-|,
  name2/.style={midway,above=-2pt}}

  \draw[ivl] (3,4) -- node[name2] {$n_{a,1}$} (4,4);
  \draw[ivl] (5,4) -- node[name2] {$n_{a,2}$} (6,4);
  \draw[ivl] (7,4) -- node[name2] {$n_{a,3}$} (8,4);
  \draw[ivl] (11,4) -- node[name2,xshift=-2pt] {$n_{a,s(a)-1}$} (12,4);
  \draw[ivl] (13,4) -- node[name2] {$n_{a,s(a)}$} (14,4);
  
  \draw[ivl] (3,1) -- node[name2] {$\ell^n_{a,0}$} (4,1);
  \draw[ivl] (3,2) -- node[name2] {$\ell^n_{a,1}$} (6,2);
  \draw[ivl] (5,1) -- node[name2] {$\ell^n_{a,2}$} (8,1);
  \draw[ivl,|-] (7,2) -- node[name2] {$\ell^n_{a,3}$} (8.5,2);
  \node at (9.5,1.5) {$\dots$};
  \draw[ivl,-|] (10.5,2) -- node[name2] {$\ell^n_{a,s(a)-2}$} (12,2);
  \draw[ivl] (11,1) -- node[name2] {$\ell^n_{a,s(a)-1}$} (14,1);
  \draw[ivl] (13,2) -- node[name2] {$\ell^n_{a,s(a)}$} (14,2);

 \end{tikzpicture}
 \end{center}
 \caption{\label{fig:red-numeral}  A numeral gadget for element $a\in A$.} 
\end{figure}
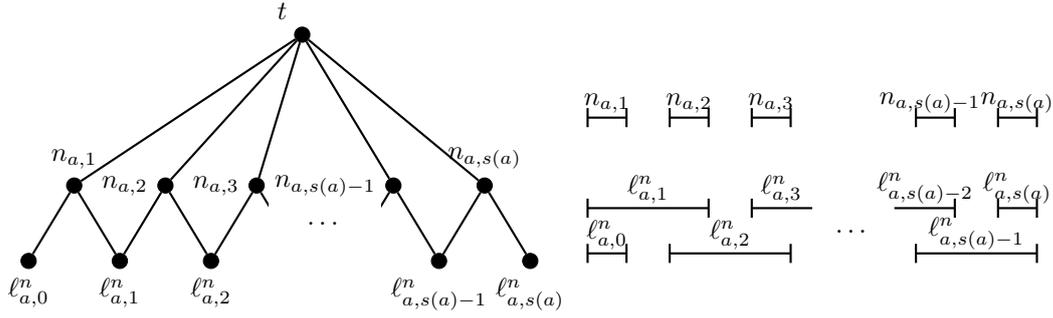

We will now show that $(b,A,s)$ is a Yes-instance for 3-\textsc{Partition} if and only if
$G$ is a $1$-interval bigraph.
For the forward direction, consider a solution $A_1, \dots, A_n$ to the 3-\textsc{Partition} instance.
We construct an interval representation following the scheme outlined in Figure \ref{fig:red-recognition},
which is missing the numeral gadgets.
Now, for each $A_i = \{x,y,z\}$, we can
intersperse the intervals $s_{i,1},\dots,s_{i,s(x)+1}$ with the intervals $n_{x,1}, \dots, n_{x,s(x)}$,
intersperse the intervals $s_{i,s(x)+1}, \dots, \linebreak[3] s_{i,s(x)+s(y)+1}$ with the intervals $n_{y,1}, \dots, n_{y,s(y)}$, and
intersperse the intervals $s_{i,s(x)+s(y)+1}, \dots, s_{i,b+1}$ with the intervals $n_{z,1}, \dots, n_{z,s(z)}$.
In this way, each vertex $\ell_{i,j}$, $1\le i\le n$, $1\le j\le b$, is non-adjacent to exactly one vertex (from $\{n_{a,1},\dots,n_{a,s(a)} : a\in A\}$) whose corresponding intervals overlap, and each vertex $\ell^{n}_{a,j}$, $a\in A$, $1\le j\le s(a)-1$, is non-adjacent to exactly one vertex (from $\{s_{i,2}, \dots, s_{i,b} : 1\le i\le n \}$) whose corresponding intervals overlap.
This certifies that $G$ is a 1-interval bigraph.

For the backward direction, we observe that our construction enforces the rigid structure from Figure~\ref{fig:red-recognition}.
Intuitively, for each $i\in \{1,\dots, n\}$, the vertices $\ell_{i,j}$ enforce an ordering of the intervals corresponding to the vertices in $S_i$, and the delimiters glue the different sections of $S_i$ vertices together in a linear fashion. Observe that between two vertices $s_{i,j}$ and $s_{i,j+1}$, we can still insert one vertex if it is adjacent to $t$, and we exploit this property to intersperse the numerals. The anchor vertices are used to stretch the structure of the slot vertices beyond the left and the right of the track $t$. This ensures then that the numerals need to be interspersed with the slots. Since there are no elements $a\in A$ with $s(a)=1$, it is also not possible for a numeral gadget to intersperse a section $S_i$ of slot vertices before $s_{i,1}$ or after $s_{i,b+1}$. In addition, the delimiters ensure that numerals do not straddle different $S_i$'s.
Therefore, we can obtain a solution to the 3-\textsc{Partition} instance by setting $A_i$ to the elements from $A$ that we used to construct the numeral gadgets that are interspersed with the slots in $S_i$.
\end{proof}

\section{Discussion}

The algorithmic challenge of CNF satisfiability and constraint satisfaction is central in both computational theory and practice, and new angles of attack to these age-old problems keep emerging. Here we focused on instances which possess a linear structure, and we proposed a new approach to dealing with local departures from such structure, as well as for deducing linear structure from partial evidence; we also identified complexity obstacles to fully exploiting and extending our approach. Our work raises several questions:

\begin{itemize}
	\item What if only one side of the bipartite incidence graph is ordered? Say we are given an ordering of variables and asked if the clauses can be inserted so as to yield a $k$-interval ordering. For the case $k = 0$ we can use the obstructions in Figure \ref{fig:tree} to guide us towards a linear ordering also of the clauses, e.g., for a pair of clauses $A, C$ with two variables $x < z$ where $xC, zA$ are edges and $zC$ is a non-edge we must place $C$ before $A$. We believe such an approach should solve the $k = 0$ case in polynomial time, but we are less optimistic about the general case of minimizing $k$.

	\item What if we are given a partial order, with some special properties, on variables and clauses? Note that already the approach for $k = 0$ hinted at above could yield a situation with a linear order on variables and a partial order on clauses.

	\item For which industrial CNF instances can we find $k$-interval orderings for low values of $k$? Our greedy algorithm for merging two linear orders to a minimum $k$-interval ordering is practical and can be applied to large instances in the SAT corpora. In light of the hardness result for recognizing 1-interval bigraphs, heuristics or domain expertise could be used to generate orders for clauses and variables, when they are not already given.

	\item Which other classes of interval bigraph CSP instances can be solved efficiently? Our hardness result is for general CSPs with large domains.
	For CSPs with Boolean domains we can show a similar hardness result albeit not for $k$-interval bigraph instances, instead for a different notion of ``imperfection'' 
	where we are given $k$ pairs of clause vertices in the incidence graph such that merging each such pair results in an interval bigraph. 
\end{itemize}

\subparagraph*{Acknowledgements}

This work was partially supported by a grant from the Peder Sather Center at UC Berkeley.
Serge Gaspers is the recipient of an Australian Research Council (ARC) Future Fellowship (FT140100048) and acknowledges support under the ARC's Discovery Projects funding scheme (DP150101134). NICTA is funded by the Australian Government through the Department of Communications and the ARC through the ICT Centre of Excellence Program.

\bibliographystyle{plain}
\bibliography{satbib}

\end{document}